\documentclass[11pt,a4paper]{article}
\usepackage{graphicx}

\usepackage{amsmath,fullpage,times}

\newtheorem{theorem}{Theorem}[section]

\newtheorem{proposition}[theorem]{Proposition}
\newtheorem{lemma}[theorem]{Lemma}

\newtheorem{corollary}[theorem]{Corollary}

\newcommand{\eps}{\varepsilon}
\newenvironment{proof}[1][Proof]{\textbf{#1.} }{\ \rule{0.5em}{0.5em}}

\begin{document}

\title{Bin stretching with migration on two
hierarchical machines}

\author{
Islam Akaria
\thanks{Department of Mathematics, University of Haifa, Haifa, Israel.\newline E--mail: islam.akaria@gmail.com.}
\and
Leah Epstein
\thanks{Department of Mathematics, University of Haifa, Haifa, Israel.\newline E--mail: lea@math.haifa.ac.il.} }

\date{}

\maketitle

\setcounter{page}{1}

\begin{abstract}
In this paper, we consider semi-online scheduling with migration on two hierarchical machines, with the purpose of minimizing the makespan. The meaning of two hierarchical machines is that one of the machines can run any job, while the other machine can only run specific jobs. Every instance also has a fixed parameter $M \geq 0$, known as the migration factor. Jobs are presented one by one. Each new job has to be assigned to a machine when it arrives, and at the same time it is possible to modify the assignment of previously assigned jobs, such that the moved jobs have a total size not exceeding $M$ times the size of the new job.

The semi-online variant studied here is called {\it bin stretching}. In this problem, the optimal makespan is provided to the scheduler in advance. This is still a non-trivial variant for any migration factor $M \geq 0$. We prove tight bounds on the competitive ratio for any migration factor $M$, where the design and analysis is split into several cases, based on the value of $M$, and the resulting competitive ratio. Unlike the online variant with migration for two hierarchical machines, this case allows an online approximation scheme.
\end{abstract}


\section{Introduction}
In this work, we study semi-online scheduling on two hierarchical machines.
Jobs are presented one by one, over a list. Each job $j$ has a positive
size (or processing time) $p_j$, and a hierarchy $g_j \in \{1,2\}$. The hierarchy corresponds to the grade of service (GoS), which this job requires. There are two
machines to be used for processing these jobs, where the speeds of the machines are unit, but the
machines are different with respect to their capabilities. The first machine $m_1$ or machine $1$
can run any job, while the second machine $m_2$ or machine $2$ cannot run jobs of hierarchy $1$, and it can only receive jobs of hierarchy $2$.
A limited amount of migration is allowed, and when a job arrives, the algorithm may reassign some jobs as long
as their total size is (at most) proportional to the processing time of
the arriving job. Specifically, there is a fixed parameter $M>0$, called {\it migration factor}, and when a job $j$ arrives, the total size of moved jobs cannot exceed $M \cdot p_j$.

The completion time (or load) of a machine is the total size of its jobs,
and the makespan is defined in a standard way as the maximum load over all machines. The goal is to
schedule the input jobs so as to minimize the maximum completion time,
that is, the goal is to minimize the makespan. We use the usual measure for online algorithms, called the competitive ratio,
for the analysis. The competitive ratio is the worst case ratio between the makespan
of an online or a semi-online algorithm and the makespan of an optimal
offline algorithm (for the same input).

Differently from the purely online version, we consider a semi-online model. In this semi-online variant,
the optimal makespan is known in advance, and
for convenience, we assume that its value is $1$, by scaling.  The variant is called bin stretching. Thus, the total size of jobs does not exceed $2$, and the total size of jobs of GoS $1$ does not exceed $1$. These simple properties are not restrictive, but they are used in the analysis of algorithms.

In this study, we consider all possible finite and positive migration factors. We show tight bounds on the competitive ratio for every possible $M>0$, which divides the problem into four cases, based on different values of the migration factor. The specific competitive ratios for the different cases are stated below.
Our algorithm will obviously schedule all jobs with
hierarchy $1$ to the first machine. In fact, our algorithms will never
migrate jobs in this case. The reason for this is that such jobs can be arbitrarily
small (and a larger job can be replaced with very small jobs, essentially without changing the input). Thus, the algorithms examine the migration option only in cases where a new job with hierarchy $2$ arrives. We study the online variant of our problem with migration on two hierarchical machines in another manuscript \cite{AE21a}. We note that here we find that the competitive ratio tends to zero as $M$ grows, in contrast to the online case with migration \cite{AE21a}, where the competitive ratio cannot be smaller than $M$, no matter how large $M$ is. In that problem, the best possible competitive ratio is $\frac 32$ for $M\geq 1$, and it is strictly larger for smaller $M$. In fact, it is exactly $\frac 53$ for $M\leq \frac 13$ \cite{JS016,YY009,AE21a}, and it is equal to $1+\frac{1}{M+1}$ for $\frac{1+\sqrt{5}}2 \leq M \leq 1$ \cite{AE21a}.

The tight
bound for online bin stretching on two hierarchical machines and the case $M=0$, which is $\frac 32$, follows from earlier work
\cite{JS016} (see also \cite{YM023}). In  \cite{JS016}, the variant where the
total size of all jobs is known in advance is studied. For this last model
it is not hard to see that any algorithm is also valid for the case where
the optimal makespan is known in advance (see below), and the lower bound on the
competitive ratio, presented in that work, is simple, and the proof is
valid for the case of known makespan (we show this later as a special
case of the case $M<\frac 12$).
Note that knowing the cost of an optimal solution allows us to design better algorithms in terms of  their competitive ratios (compared to the case where the makespan may be arbitrary), but it does not simplify the design of an algorithm, and often the design becomes more advanced.

The two similar semi-online problems, bin stretching \cite{AzarR01,BBSSV17a,BBSSV17b,GBK17,KK13,GKB15}, and scheduling with known total size \cite{KKST97,ANST04,CKK05,AH12,KKG15}, were both studied, in particular, for identical machines . It is not hard to see that knowing the makespan is a stronger assumption compared to knowing the total size, so an algorithm for the latter can be used for the former (with the same competitive ratio), and a lower bound construction for the competitive ratio of the former can be used for the latter. The issue of the relation between the two models was addressed specifically \cite{LeeL13}. For two identical machines, there is essentially no difference between the two variants \cite{AzarR01,KKST97}, but for a general number of machines is was proved that the two models are different \cite{AH12,BBSSV17a}. It follows from previous work \cite{JS016,YM023} that for two hierarchical machines (without migration) the two variants are also very similar.
In the last section we address the same question for the problem studied here, and show that the competitive ratio for any migration factor $M$ will be bounded away from $1$, that is, it will not get closer to $1$ as $M$ grows, as it is the case for bin stretching.

Online scheduling with migration, where it is allowed to migrate jobs whose total size is proportional to the
size of an arriving job, was first proposed by Sanders et al. \cite{PN-2009}, who studied the problem on identical machines (without hierarchies). That work contains in particular a linear time {\it online approximation scheme,}
which is a family of online algorithms with migration, such that the competitive ratio can be arbitrarily close to
$1$ for sufficiently larger migration factors. The required migration factor increases as the competitive ratio is closer to $1$, and it is exponential in the value $\frac 1{\eps}$, where $\eps$ is the competitive ratio minus $1$.
For two identical speed machines, it is known that the required migration factor is polynomial in $\frac{1}{\eps}$ \cite{IW}, and given the results here, one can also obtain such a scheme for bin stretching and two hierarchical machines. However, as explained above, this is possible since the optimal makespan is known in advance, and not possible for online algorithms with migration \cite{AE21a}.
Job migration in scheduling according to
this migration model has been studied further for other variants of scheduling and other combinatorial optimization problems \cite{SV-2016,EL-2014,EpsteinL09,EL19,GSV18,BJK20,Levin22}.

As we explain above, the case of identical machines allows an online approximation scheme for any $m$, and this is not the case for hierarchical machines. Moreover, the lower bound will hold for any $m\geq 2$ by constructing the input used for two machines without defining any jobs that can run on machines with indices larger than $2$. However, one can still improve the result without migration by using migration. For the more general case of restricted assignment (where every job has a subset of machines where it can run), migration is not helpful at all for two machines \cite{AE21a}. Here, we show not only that bin stretching for two hierarchical machines has an online approximation scheme, but also that the migration factor is polynomial in $\frac 1{\eps}$. As explained above, we show that such a scheme is not possible for the similar variant of known total size of jobs.

The online hierarchical scheduling problem for multiple machines, or scheduling with grades
of service (GoS) was first proposed by Bar-Noy et al. \cite{BNFN01}, where an algorithm whose
competitive ratio is a constant was designed. Each job, as well as each machine, has a hierarchy
associated with it. A job can be scheduled on a machine only when its hierarchy is no higher than
that of the machine. There are further studies of different hierarchical variants online and semi-online scheduling. The work of Park et al.~\cite{JS016}, mentioned above, contained a study of the case of two machines for online and semi-online algorithms, while the work of Jiang  et al.~\cite{YY009} contained a study of the case of two machines for online algorithms (including a preemptive variant). Wu et al. \cite{YM023} investigated semi-online versions for two hierarchical machines. Specifically, they showed tight bounds of $(1+\sqrt{5})/2$ for the case where the largest processing time of any job is
known in advance.
There are also multiple studies of the online hierarchical scheduling problem for parallel machines,
its special cases where there cannot be a large number of hierarchies \cite{ZJT09,CG04,Jiang08,TA11,LLC11,LHL2014}, and semi-online variants with known total size of jobs \cite{LHL2014}. There is a fair number of articles focusing on semi-online hierarchical scheduling on two machines. Xiao et al. \cite{XWL19} investigated the
problem in several cases (including the total size of low-hierarchy is known, and the case where the total size of
each hierarchy is known, see also \cite{LX14,XZ015}). There is work on other semi-online models and hierarchical machines, such as bounded sizes and combined information \cite{LCXZ11,LX16,ZY015}. There is also work for
similar variants with different objectives \cite{QY19,LX15}.
 It is worth noting that such models of semi-online scheduling were also studied for other scheduling problems without hierarchies (see for example
\cite{HZ99,ZY019,GY003,XL015}).

The paper is organized as follows. We provide definitions and notation in the second section.
In the other sections, we discuss the problem and split it into four cases with respect to the value of
the migration factor $M$ (see Figure \ref{f:5} for a specification of the bounds). Obviously, the competitive ratio as a function of $M$ is monotonically non-increasing.
We design and analyze four algorithms, and we also prove a matching lower bound on the competitive ratio for each case, obtaining tight bounds for all finite values of $M$. The first case is $M\geq \frac 52$, for which the tight bound on the competitive ratio as a function of $M$ is $\frac{2M+5}{2M+3}$. In the second case, $\frac 34 \leq M < \frac 52$, the tight bound is on the
competitive ratio is $1.25$, achieved by an algorithm whose migration factor does not exceed $\frac 34$, while the lower bound of $1.25$ on
the competitive ratio will hold for any semi-online algorithm with $M \leq 2.5$. In the third case, $\frac 12 \leq M < \frac 34$, we prove
a tight bound of $2-M$ by presenting two algorithms defined over two different domains in the interval $[\frac 12,\frac 34)$,
and proving a lower bound on the competitive ratio for any algorithm and the suitable value of $M$. In the last case, $0 \leq M < \frac 12$, we show that this
case is equivalent to the case without migration, by presenting a lower bound of $\frac 32$ on the competitive ratio for any online algorithm
whose migration factor is below $\frac 12$. The idea for the lower bound is similar to that which was presented in \cite{JS016}, and the algorithm presented in that work (with migration factor $M=0$) can be used, since this is an algorithm for known total size. At the end of
the fourth case, we show that if replacing the assumption of known makepsan with the
assumption of known total size, the problem would have been different, unlike the case $M=0$, for which the two problems are known to be similar \cite{JS016,YM023}.

\section{Preliminaries}

Throughout the paper we will use the following notation. Every job will be denoted by its index in $\{1,2,\ldots,n\}$ and presented to the online
algorithm in this order, where the number of jobs $n$ is not known in advance. Each job $j$ (also denoted by $J_j$) is an ordered pair $J_j=(p_j,g_j)$, where $p_j>0$ is the processing time and $g_j\in \{1,2\}$ is the grade of service (GoS or hierarchy) of job $j$. A job with hierarchy $g_j$ can be assigned to a machine in $\{1,\ldots,g_j\}$.

The set of jobs with GoS $1$ is denoted by $X$, and as mentioned above, these jobs are always assigned to machine $1$.
We let $Y$ be the set of jobs with GoS $2$, that are assigned to the second machine by a fixed algorithm at a
certain time. Similarly, the set of jobs with GoS $2$ that are scheduled on the first machine at a certain time (by the same algorithm) is denoted by $Z$.
Note that jobs of hierarchy $2$, which are assigned to a machine at a certain time, could have been scheduled to that machine by the algorithm at its arrival time, or they might have been assigned to the other machine, but they were migrated to this machine at a later time. Any job may be migrated multiple times.
Here we will use an additional variable, $W$. This variable denotes a subset of jobs with GoS $2$ selected by the algorithm, where the algorithm may migrate it or it may keep it on the second machine.

We let $x_j$, $y_j$, $z_j$, and $w_j$ denote the total sizes of the jobs of $X$, $Y$, $Z$, and $W$, respectively, at time $j$, that is, just after job $j$ was scheduled (and all migrations corresponding to the arrival of $j$ have been performed). We let $x_0=y_0=z_0=w_0=0$, since the sets are empty before any job was presented to the algorithm.

Thus, the loads of $m_1$ and $m_2$ after $j$ was assigned are $x_j+z_j$ and $y_j$, respectively. We sometimes let $X_j$, $Y_j$, $Z_j$, and $W_j$ denote the sets $X$, $Y$, $Z$, and $W$ (respectively) just after $j$ was assigned, for clarity. The maximum processing time and the second maximum processing time out of the processing times of the jobs of $Y_{j-1}$ are denoted by $p^{\max{Y}}_j$ and $p^{\max{Y,2}}_j$, and the suitable jobs are denoted by $j^{\max Y},\ j^{\max Y,2}$, respectively. Each of these values is defined to be zero if it is not well-defined (which happens in the case $|Y_{j-1}|\leq 1$).

We let $T_1$ and $T_2$ be the completion times (or loads) at the end of the input of the first and the second machine respectively, i.e. the makespan equals $\max\{T_1,\ T_2\}$. We also let $OPT_j$ denote the makespan of an optimal offline solution just after job $j$ was scheduled. For a complete input with $n$ jobs, we have $T_1=x_n+z_n$ and $T_2=y_n$. For any input $I$, let $c^*(I),\ c_{Alg}(I)$ denote the makespan of an optimal solution and the solution of algorithm Alg, respectively. In particular, we have $c^*(I)=OPT_n$.
The competitive ratio for input $I$ is therefore $\frac{c_{Alg}(I)}{c^*(I)}$, and the competitive ratio of $Alg$ is the supremum of these values over all possible instances $I$.


\section{The case $\boldsymbol{M \geq 2.5}$}

In this section, we prove a tight bound on the competitive ratio for the semi-online problem and the case $M\geq 2.5$. Let us define a value $\mu$, which is a function of $M$, by $\mu = \frac{2}{2M+3}$, where $0<\mu \leq \frac 14$ (and $\frac{1}{\mu}\geq 4$). The tight bound on the competitive ratio for this problem is equal to $1+\mu=1+\frac{2}{2M+3}$, where this value tends to $1$ for $M$ growing to infinity, and it is never exceeds $\frac 54$. The algorithm whose competitive ratio is $1+\mu$ is based on keeping the final completion time of the second machine in the interval $[1-\mu,1+\mu]$. It is obvious that if this is achieved, the completion time for the first machine will not exceed $1+\mu$, because as mentioned earlier, the sum of all jobs is not greater than $2$. After we complete the analysis of our algorithm, we will prove a lower bound on the competitive ratio for every algorithm, that is equal to $1+\mu$. As a result, we get a tight bound of $1+\mu$ on the competitive ratio for this case, where $\mu$ is the above function of $M$.

\subsection{An algorithm}

\textbf{Algorithm \textit{A}}

Let $X=\emptyset,\ Y=\emptyset,\ Z=\emptyset;$

Repeat until all jobs have been assigned:
\begin{enumerate}
\item Receive job $j$ with $p_j$ and $g_j$; \item If one of $g_j=1$ and
$y_{j-1}\geq 1-\mu$ holds, schedule $j$ on the first machine and
update: $X \leftarrow X\cup \{j\}$, or $Z \leftarrow Z\cup \{j\},$ respectively.

            return to step 1.
\item If $y_{j-1}+p_j \leq 1+\mu$ holds, schedule $j$ on the second machine
 and update: $Y\leftarrow Y\cup \{j\}$,

            return to step 1.
\item Let $W$ be a subset of $Z \cup Y\cup \{j\}$ of maximum total
size not exceeding $1$. Update the
schedule such that all jobs of $W$ are assigned to the second machine, and all other jobs
are assigned to the first machine. Update $Y \leftarrow W,\ Z
\leftarrow \{Z \cup Y\cup \{j\}\} \setminus W$, return to step 1.
\end{enumerate}

The algorithm applies a procedure that tries to find a subset $W$ with a compatible size to maintain the balance of the machines, such that
the completion times will not exceed $1+\mu$ (in fact, they will not exceed $1$ in this case). Steps $2$ and $3$ deal with simple cases, where the new job has GoS $1$, or the job can be assigned without any migration such that no machine will have completion time above $1+\mu$. For the cases with a job with a grade of service $2$, this will be obvious, while the case where the grade of service of the new job is $1$ will be analyzed. Step $4$ is the case where jobs of GoS $2$ should be rearranged, such that the schedule is more similar to that of the current optimal solution.

Note that the running time of this algorithm is exponential due to the last step where the subset sum problem is solved. By applying an approximation scheme for this problem rather than an exact algorithm we can obtain a polynomial running time at the expense of a slightly larger competitive ratio. Now, we prove an upper bound on the competitive ratio, and analyze its migration factor.

Consider a fixed optimal offline solution, and let $o_{ji}$ denote the load of $m_i$ in this solution for $i=1,2$, after $j$ jobs were assigned. We have $o_{ji}\leq 1$ since the final optimal makespan is $1$.

\begin{lemma}\label{lemo1}
After $j$ jobs have been assigned by the algorithm, the load of the second machine
$y_j$ is at least $\min\{1-\mu,o_{j2}\}$.
\end{lemma}
\begin{proof}
Assume by contradiction that there is an index $j$ for which $y_{j} < \min\{1-\mu,o_{j2}\}$, and that $j$ is the minimum index for which this holds. Consider the most recent time before time $j$, $j'\leq j$, that a job of GoS $2$ was assigned by the algorithm not by applying step 3, if such a time exists.

If there is no such $j'$, this means that for jobs $\{1,2,\ldots,j\}$ it holds that all jobs of GoS $2$ are assigned to machine $m_2$ (this set may contain any number of jobs, and it could be empty), while machine $m_1$ has exactly the jobs of GoS $1$. Machine $m_2$ has a subset of jobs with GoS $2$ in any solution including the optimal one, while the algorithm assigns all these jobs to $m_2$, and therefore $y_{j} \geq o_{j2} \geq \min\{1-\mu,o_{j2}\}$.

If there was an assignment of a job $j' \leq j$ whose grade of service is $2$ by step $2$, and there was no assignment by step $4$ up to and including job $j$, this means that $y_{j'-1}\geq 1-\mu$, and the load of $m_2$ could not decrease until time $j$ (since step $4$ is not applied, there is no migration and loads are monotonically non-decreasing as a function of $j$), so $y_{j}\geq y_{j'-1} \geq 1-\mu$.

Finally, assume that step $4$ was applied for job $j'\leq j$, and afterwards no job with GoS $2$ was assigned by step $2$ up to and including time $j$. Machine $m_1$ received a set of jobs of GoS $2$ in the application of step $4$ for job $j'$. When step $4$ is performed, machine $m_2$ receives a  set of jobs of maximum size not exceeding $1$ out of all possible subsets of already existing jobs. Since the optimal makespan at termination is $1$, the optimal makespan at this time is also at most $1$, and therefore $m_2$ has a set of jobs of total size at most $1$ in optimal solutions, that is, $o_{j'2} \leq 1$, and since one of the options for $W$ is the set of jobs of $m_2$ in the optimal solution at this time, we find $y_{j'}=w_{j'} \geq o_{j'2}$. If $j'=j$, we are done (since we already reach a contradiction), and otherwise,
according to the assumption we have that $y_{j} < \min\{1-\mu,o_{j2}\}$. By the choice of $j'$, all jobs in $\{j'+1,\ldots,j\}$ whose GoS is $2$ are assigned by step $3$. Thus, the jobs assigned to machine $m_2$ at time $j$ (the set $Y_j$) are exactly those of $W_j$ together with all jobs of GoS $2$ among $\{j'+1,\ldots,j\}$. Let the total size of this last set of jobs be denoted by $\lambda \geq 0$. Based on the optimal solution, we find that among jobs $1,2,\ldots,j'$ there is a subset of jobs of GoS $2$ whose total size is at least $o_{j2}-\lambda$, where $o_{j2}-\lambda \leq 1$ since $o_{j2} \leq 1$ and $\lambda\geq 0$. This subset could have been chosen as $W_{j'}$ and therefore
$w_{j'} \geq o_{j2}-\lambda$. However, the current load  of $m_2$ for the algorithm after $j$ is assigned satisfies $y_j = w_{j'}+\lambda$. Combining with $y_j < o_{j2}$ gives $o_{j2} > y_j = w_{j'}+\lambda \geq o_{j2}$, a contradiction.
\end{proof}

\begin{corollary}
After $j \geq 1$ jobs have been assigned, the load of the first machine is at most $1+\mu$.
\end{corollary}
\begin{proof}
By Lemma \ref{lemo1}, we have $y_{j} \geq \min\{1-\mu,o_{j2}\}$. We also have $x_j+y_j+z_j=o_{j1}+o_{j2}$, which is the total size of the first $j$ jobs. We bound the load of $m_1$, which is $x_j+z_j$ as follows. If $y_j \geq 1-\mu$, using $x_j+y_j+z_j \leq 2$, we get $x_j+z_j \leq 2-y_j \leq 1+\mu$. If $y_j \geq o_{j2}$, we get $x_j+z_j =o_{j1}+o_{j2}-y_j \leq o_{j1} \leq 1$, since an optimal solution for a sub-input has makespan not exceeding $1$.
\end{proof}

\begin{lemma}
After $j$ jobs have been assigned, the load of the second machine is at most $1+\mu$.
\end{lemma}
\begin{proof}
We prove the claim by induction. Obviously, it holds for $j=0$ since $y_0=0$. If $y_j=y_{j-1}$, we are done. Otherwise, one of steps  $3$ and $4$ was applied. If $j$ is assigned in step $3$, then we have $y_j=y_{j-1}+p_j\leq 1+\mu$. Otherwise, $y_j=w_j \leq 1$.
\end{proof}

\begin{lemma}
For every $M\geq 2.5$, the migration factor of the algorithm is at most $M=\frac 1{\mu}-\frac 32$.
\end{lemma}
\begin{proof}
The only step where jobs may be migrated is step $4$, and in this case the GoS of the new job $j$ is $2$ ($g_j=2$). Note that $j$ was not assigned to any machine prior to this step and therefore it is not migrated but only other jobs of $W_j$ may be migrated.

We bound the loads of the two machines before and after the reassignment. By the action of the algorithm, $y_j=w_j\leq 1$, and $w_j \geq o_{j2}$. Moreover, by the choice of $W_j$ such that $w_j \geq o_{j2}$, we find that the resulting load of $m_1$ will be at most $o_{j1}+o_{j2}-y_j=o_{j1}+o_{j2}-w_j \leq o_{j1}\leq 1$. Thus, both loads will be at most $1$ for the algorithm after the assignment of $j$ is completed.

Before the application of step $4$, since step $3$ is not applied, it holds that $y_{j-1}+p_j >1+\mu$, and  we have $x_{j-1}+z_{j-1} \leq 2 - y_{j-1}-p_j \leq 1-\mu$. Since step $2$ was not applied and $g_j=2$, it was the case that $y_{j-1} < 1-\mu$, so we have that both loads were at most $1-\mu$. We note the case where one of the machines already has a total size of at least $1-\mu$ is in fact easy in the process of assigning a job of GoS $2$.

Let $i$ be the machine that has $j$ after the reassignment. The load of machine $i$ will be at most $1$, and the total size of jobs migrated from $m_{3-i}$ is at most $1-p_j$. Since the load of $m_i$ before the migration is at most $1-\mu$, this is an upper bound on the total size of jobs migrating to $m_{3-i}$. We find that the total size of migrated jobs does not exceed $2-p_j-\mu$.

We bound the value $p_j$ from below. Since step 2 was not applied, we have $y_{j-1}+p_j>1+\mu$. Using $y_{j-1}<1-\mu$ as well, we have $p_j > 1+\mu-y_{j-1} > (1+\mu)-(1-\mu)=2\cdot \mu$. Thus, the migration factor for the step of assigning $j$ is below $\frac{2-p_j-\mu}{p_j} <
\frac{2-3\mu}{2\mu}=\frac 1{\mu}-\frac 32$.
\end{proof}

We conclude with the following theorem.
\begin{theorem}
The competitive ratio for Algorithm $B$ is at most $1+\mu=1+\frac{2}{2M+3}=\frac{2M+5}{2M+3}$, and its
migration factor is at most $M$.
\end{theorem}
\begin{proof}
The proof follows directly from the lemmas and corollary above.
\end{proof}

\subsection{A lower bound on the competitive ratio of the case where $\boldsymbol{M\geq 2.5}$}

\begin{theorem}
\label{ALGALOW2.5} The competitive ratio for every algorithm when $M \geq 2.5$ is at least $1+\mu = 1+\frac{2}{2M+3}$.
\end{theorem}

\begin{figure}[h!]
\begin{center}
\begin{tabular}{c}
\includegraphics[width=160mm]{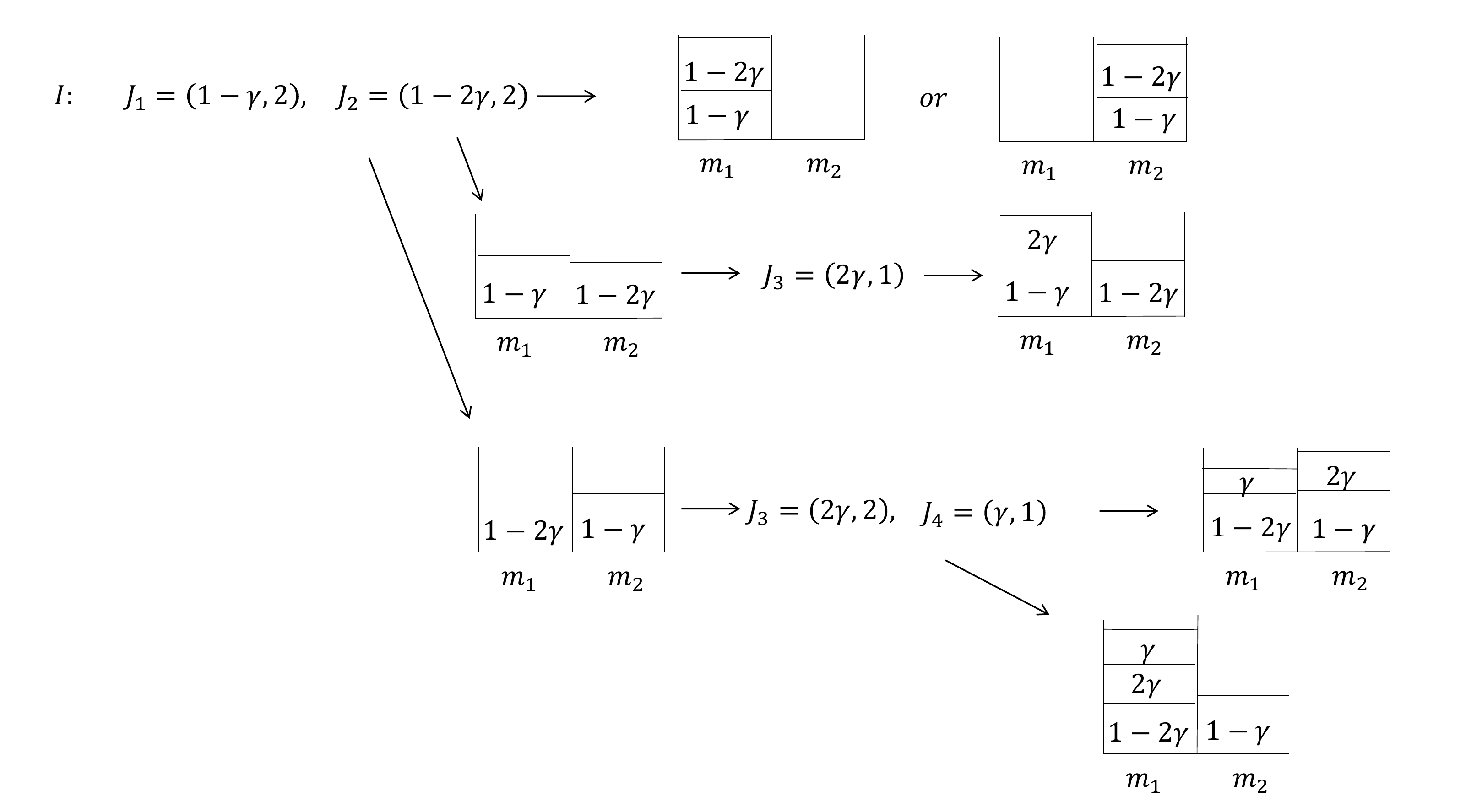}
\end{tabular}
\caption{\small
The schedules produced by an online algorithm with migration factor $ M \geq 2.5$ in the main cases of the proof of Theorem \ref{ALGALOW2.5}, represented as a decision tree. Some of the cases where both large jobs are assigned to the same machine are omitted.}
\label{f:1}
\end{center}
\end{figure}

\begin{proof} We define an input based on the actions of a fixed online algorithm. Let the first two jobs in the input be $J_1=(1-\gamma,2)$ and $J_2=(1-2\gamma,2)$, where $\gamma$ is a constant such that $ 0 < \gamma < \mu \leq \frac 14$ (whose value is close to $\mu$).

Consider the schedule of the online algorithm after the first two jobs were presented and all migrations were performed. If both jobs are
assigned to the same machine by the algorithm at this time, the makespan is $2-3\gamma \geq \frac 54 \geq 1+\mu$. In this case there are also very small jobs of GoS $2$ and total size $3\gamma$ (for example, six jobs of size $\frac{\gamma}2$), so the makespan of an optimal solution is indeed $1$, and the first two jobs cannot be migrated.

If the algorithm scheduled the first job on the first machine, and it scheduled the second job on the second machine, the input proceeds with a third job $J_3=(2\gamma,1)$. The optimal cost now is exactly $1$. This job cannot cause the migration of both previous jobs simultaneously as $\frac{2-3\gamma}{2\gamma}=\frac{1}{\gamma}-\frac 32 > \frac{1}{\mu}-\frac 32=M$. If it causes the migration of one
existing job, there will be a machine of load $2-3\gamma >1+\mu$, and the situation is as in the case where the two jobs were assigned to the same machine. Finally, if there is no migration, the first machine has load $1+\gamma$, as $J_3$ cannot be assigned to $m_2$. As mentioned earlier, an optimal solution has makespan $1$, since it swaps the machines of the first two jobs.

If the algorithm schedules the first job on the second machine, and the second job on the first machine, we continue to the next jobs in the input that are $J_3=(2\gamma,2), J_4=(\gamma,1)$. Once again, none of these jobs can cause the migration of more than one job out of the first two jobs, and the migration of just one job results in a makespan above $1+\gamma$. Consider the schedule after termination. Machine $m_1$ has jobs $J_2$ and $J_4$. If it also has $J_3$, its load is $1+\gamma$. Otherwise, $m_2$ has $J_1$ and $J_3$, and its load is $1+\gamma$, so the makespan of the algorithm is $1+\gamma$ in both cases, An optimal solution can schedule $J_1$ and $J_4$ on $m_1$ and the other two jobs on $m_2$.

Because all optimal schedules for the above inputs have makespan $1$, while the algorithm always has makespan not below $1+\gamma$, by letting $\gamma$ tend to $\mu$, we get a lower bound of $1+\mu$ on the competitive ratio of any online algorithm.

The  diagram of Figure \ref{f:1} describes the difficult cases of the aforementioned input.
\end{proof}

\section{The case $\boldsymbol{\frac 34\leq M<\frac 52}$}

In this section, we prove a tight bound on the competitive ratio for the same semi-online problem, but with a different values for the migration factor $M$, which will now satisfy $\frac 34 \leq M< \frac 52$. The tight bound on the competitive ratio for this problem is equal to $1.25$. The proof of a tight bound for this problem in this case is divided into two parts as in the previous case. In the first part we present an $Algorithm \; B$ whose competitive ratio is $1.25$ and its migration factor does not exceed $\frac 34$. The second part is a
lower bound of $1.25$ on the competitive ratio for any online algorithm with $M \leq 2.5$. This part was in fact proved in the previous section, since this is the lower bound for the case $M=2.5$.

The idea behind this algorithm is similar to that of $Algorithm \ A$, except that $Algorithm \ B$ maintains the load of machine $m_2$ such that it will not exceed $1.25$, and such that the final load will be within a fixed interval $[0.75,1.25]$. Due to this property, the final load of machine $m_1$ will not exceed $1.25$. Now, we present the algorithm and analyze it, and prove an upper bound on the competitive ratio.

\subsection{An algorithm}

\textbf{Algorithm \textit{B}}

Let $X=\emptyset,\ Y=\emptyset,\ Z=\emptyset,\ W=\emptyset;$

Repeat until all jobs have been assigned:
\begin{enumerate}
\item Receive job $j$ with $p_j$ and $g_j$; \item If at least one of $g_j=1$ and $y_{j-1} \geq 0.75$ holds,
 assign $j$ to the first machine and update: $X \leftarrow X\cup \{j\},$ or $Z \leftarrow Z\cup \{j\},$ respectively.

            return to step 1.
\item If $y_{j-1}+p_j \leq 1.25$ holds, assign $j$ to the second machine and update: $Y\leftarrow Y\cup \{j\},$

            return to step 1. In the remaining cases $Y_{j-1}$ is non-empty.
\item If $p_j \geq 0.75$ holds,
sort the jobs in $Y$ by non-increasing size, and update $W$ to be the maximum prefix of the sorted list with total size at most $0.75 \cdot p_j$ (which may be empty).
\begin{itemize} 
\item If $y_{j-1}-w_j+p_j > 1.25$ holds, assign $j$ to the first machine and update: $Z\leftarrow Z\cup\{j\}$.
\item Otherwise, migrate the jobs of $W$ to the first machine, assign $j$ to the second machine and update: $Y\leftarrow (Y \setminus W) \cup \{j\}$ and $Z\leftarrow Z\cup W$.\end{itemize}
            return to step 1.
\item In this case $p_j + p^{\max Y}_j > 1.25$ holds. Assign $j$ to the first machine and update: $Z\leftarrow Z\cup \{j\}$.

Otherwise:
\begin{itemize}
\item If $p^{\max Y}_j \geq \frac{y_{j-1}}{2}$ holds, update $W\leftarrow Y \setminus \{j^{\max Y}\}$.
\item If $0.25 \leq p^{\max Y}_j < \frac{y_{j-1}}{2}$ holds, update $W\leftarrow \{j^{\max Y}\}$.
\item In this case, $p^{\max Y}_j < 0.25$ holds. Sort the jobs of $Y$ by non-increasing size and update $W$ to be the minimum length prefix of the sorted list with total size at least 0.25, and the entire list  $Y$ if no such prefix exists.\\
If $w_j > 0.75 \cdot p_j$ holds, update $W\leftarrow Y \setminus W$.\end{itemize}
Migrate all jobs of $W$ to the first machine, assign $j$ to the second machine and update: $Z\leftarrow Z\cup W$ and $Y\leftarrow (Y \setminus W)\cup \{j\}$.

    return to step 1.
\end{enumerate}

The idea of the first three steps in $Algorithm \ B$ is the same idea of the $Algorithm \ A$. In the last steps, step $4$ and $5$, when $Algorithm \ B$ is forced to migrate jobs, the idea is different. In these steps the algorithm tries to find a subset $W$ of $Y$ (in polynomial time), such that it can migrate this set to machine $m_1$. Jobs are only ever migrated from machine $m_2$ to machine $m_1$. The algorithm attempts to find such a subset of jobs that their total size satisfies the following property for $w_j$: $\gamma_j \leq w_j \leq 0.75 \cdot p_j$, where we let $\gamma_j= y_{j-1}+p_j-1.25$. The value  $\gamma_j$ is a lower bound on the total size that has to be migrated from the second machine to obtain $y_j\leq 1.25$ in order to allow $j$ to be scheduled on the second machine. If step $4$ is reached, by the condition of step $3$ (which does not hold), $\gamma_j$ is strictly positive (and so is $y_{j-1}$ which is implied by $p_j\leq 1$).
The way $W$ is selected is based on the required migration factor of $0.75$, and the correctness is based in particular on the  fact that in each input there are no three jobs whose processing time of each pair of them is greater than $1$, since the optimal cost is $1$ (by the pigeonhole principle). In both steps $4$ and $5$, if the algorithm does not find a set $W$ as required, then job $j$ will be assigned to machine $m_1$ (which is tested in advance in both steps). In what follows, we prove that for each input, once the algorithm enters step $4$ or step $5$, it will never enter any of those steps again. In both cases (no matter whether there is a suitable subset $W$ or not) the final load of the two machines will be at most $1.25$.

In the next lemma we discuss step $4$, and the case where the new job is assigned to the first machine, which is done since the condition $y_{j-1}-w_j+p_j > 1.25$ holds. In this case, $y_{j-1}>0$ holds, and therefore $Y$ is non-empty and $p^{\max Y}_j>0$. If $W$ is non-empty, it holds that $0 < p^{\max Y}_j \leq 0.75 \cdot p_j$.

\begin{lemma} \label{lB:1}
Consider step $4$. If $j$ is assigned to the first machine and $W$ is not empty, then it holds that $p_j + p^{\max Y,2}_j > 1$ and $p^{\max Y}_j + p^{\max Y,2}_j > 0.75\cdot p_j > 0.5$  (and $p_j+p^{\max{Y}}_j > 1$ holds as well).
\end{lemma}
\begin{proof}
We start with the properties which hold for any job $j$ assigned to machine $m_1$ in step $4$. Since this step is applied, we have $p_j \geq 0.75$. Since step $2$ was not applied, we have $y_{j-1}<0.75$. Since step $3$ was not applied, we have $y_{j-1}+p_j>1.25$.

By $y_{j-1}-w_j+p_j > 1.25$, which is the condition for assignment to the first machine, we get $w_j<y_{j-1}+p_j-1.25$. By $y_{j-1}<0.75$, this yields $w_j<p_j-0.5$, and by $p_j \leq 1$, we find $w_j<p_j-\frac{p_j}2=\frac{p_j}2$. Since $p_j\leq 1$, we also have $y_{j-1}-w_j > 0.25$, so there is at least one job of $Y_{j-1}$ that is not included in $W_j$. In particular, $p^{\max Y,2}_j$ is well-defined, since $W$ is non-empty.

Let $J_k$ be the first job (with the maximum processing time) in $Y \setminus W$ according to the sorted list of $Y_{j-1}$ (which was used to find $W_j$). Let $\alpha \geq 2$ be the index of $k$ in the list for $Y_{j-1}$ (it holds that $\alpha \neq 1$, since $W$ is non-empty). Every job among the first $\alpha$ jobs of this list has size at least $p_k$, so $w_j \geq (\alpha-1)\cdot p_k$. Consider the case $\alpha \geq 3$, where $W_j$ consists of at least two jobs, and $w_j \geq 2\cdot p_k$ (or equivalently, $p_k\leq \frac{w_j}2$). In this case, we have $w_j+p_k \leq \frac{3\cdot w_j}2$ and $w_j+p_k>0.75 \cdot p_j$ (by the choice of $W$), yielding $w_j>\frac{p_j}2$, a contradiction to a property proved earlier. Thus, we are left with the case $\alpha=2$. In particular, $p_k=p^{\max Y,2}_j$, and by $w_j+p_k>0.75 \cdot p_j$ we obtain that $p^{\max Y}_j + p^{\max Y,2}_j > 0.75\cdot p_j > 0.5$, since $p_j \geq 0.75$.

Since $w_j=p^{\max Y}_j$, we get $p^{\max Y}_j<p_j-0.5$, or alternatively, $p_j-p^{\max Y}_j>0.5$, and we also use  $p^{\max Y}_j+p^{\max Y,2}_j>0.5$, where we already proved the last inequality. Taking the sum of the two inequalities provides the first property of the lemma. It also holds that $p_j+p^{\max{Y}}_j>1$, because $p^{\max{Y}}_j\geq p^{\max{Y,2}}_j$.
\end{proof}

\begin{lemma} \label{lB:2}
For any job $j$ satisfying $p_j\leq 0.5$, $j$ is assigned in step $2$ or step $3$. Once the property $y_j \geq 0.75$ holds for some job $j$, all jobs arriving after $j$ will be scheduled in step $2$.
\end{lemma}
\begin{proof}
We start with the first part. If $g_j=1$, $j$ is assigned in step $2$. If prior to the arrival of $j$ it holds that $y_{j-1} \geq 0.75$, it is also assigned in step $2$. Thus, we are left with the case $g_j=2$ and $y_{j-1}<0.75$. In this case, $y_{j-1}+p_j < 1.25$, and therefore $j$ is assigned in step $3$.

The second part holds since once $y_j \geq 0.75$ holds, it can be shown inductively that every job is assigned in step $2$ (to $m_1$), and the property will hold after the assignment.
\end{proof}

\begin{lemma}\label{lB:3}
For every input, the algorithm enters step $4$ at most once. And if the algorithm enters step $4$, it will assign all future jobs in steps $2$ and $3$.
\end{lemma}
\begin{proof}
If step $4$ was not applied at all during the execution, we are done. Otherwise, let $j$ be the first job in which the algorithm enters step $4$, where in particular, we have $p_j \geq 0.75$. Since step $4$ is reached, step $3$ was not applied and $y_{j-1}+p_j>1.25$ holds. Recall that in this case $Y_{j-1}$ is non-empty. If $j$ is assigned to machine $m_1$ and $W$ is empty, the first job in the sorted order satisfies $p^{\max Y}_j > 0.75 \cdot p_j>0.5$. Since both $j$ and an earlier jobs have sizes above $\frac 12$, two jobs of sizes above $0.5$ have arrived already. The input can contain at most two such jobs, so any job that arrives after $j$ has a processing time of at most $0.5$, and by the previous lemma, every such job is assigned in step $2$ or step $3$.

If $j$ is assigned to machine $m_1$ and $W$ is not empty, then by the second property of Lemma \ref{lB:1}, the set $Y_{j-1}$ contains two jobs whose total processing time is greater than $0.75 \cdot p_j > 0.5$. None of the two jobs can be scheduled on the same machine with $j$ in any optimal solution  by the first property of Lemma \ref{lB:1}. Any optimal solution has one machine with $j$, whose size is above $\frac 12$, and another machine with two jobs that arrived before $j$, where their total size is also above $\frac 12$. Thus, so each job that arrives after $j$ has processing time less than $0.5$ (any job arrives after $j$ should be with scheduled with $j$ or with the above two largest jobs of $Y_{j-1}$ in any optimal solution), and once again all further jobs are scheduled in a step prior to step $4$, by the first part of Lemma \ref{lB:2}.

Finally, if $j$ is assigned to machine $m_2$, no matter what the properties of $W$ are, we get $y_j \geq p_j \geq 0.75$, and by the second part of Lemma \ref{lB:2}, this means that all jobs will be assigned in step $2$.
\end{proof}

In what follows, let $\Lambda$ denote the total size of jobs of GoS $1$ in the entire input.

\begin{lemma} \label{lB:4}
If step $5$ is applied for a new job $j$, and it holds that $p_j + p^{\max Y}_j \leq 1.25$, then the algorithm always manages to update $W$ such that $\gamma_j \leq w_j \leq 0.75 \cdot p_j$ and $w_j<0.5$, and after assigning $j$, the makespan will not exceed $1.25$. If afterwards $m_1$ receives only jobs of GoS $1$, then its load will remain at most $1.25$.
\end{lemma}
\begin{proof}
Consider job $j$ that is assigned in step $5$. Since $j$ is not assigned in an earlier step, and due to the assignment in step $5$, we have $0.5< p_j <0.75$, $y_{j-1}<0.75$, and $p_j+y_{j-1} > 1.25$. So $0.5<y_{j-1}<0.75$ holds as well. In particular, $Y_{j-1}$ is non-empty.
Additionally, we have $\gamma_j= y_{j-1}+p_j-1.25 <0.75 +0.75-1.25= 0.25$. Suppose that $p_j + p^{\max Y}_j \leq 1.25$ holds. We first show that in all three options of case $5$ for assigning $j$ to the second machine we have $\gamma_j \leq w_j \leq 0.75 \cdot p_j$ and $w_j<0.5$.

Now, if $p^{\max Y}_j \geq \frac{y_{j-1}}{2}$ holds, the algorithm updates $W$ to be $Y \setminus \{j^{\max Y}\}$, and in this case we get
$w_j=y_{j-1}-p^{\max Y}_j \leq \frac{y_{j-1}}{2} < \frac38 < 0.75\cdot p_j$. On the other hand $\gamma_j= y_{j-1}+p_j-1.25 =  (w_j + p^{\max Y}_j) +p_j-1.25$, since $p^{\max Y}_j +p_j \leq 1.25$, this yields $\gamma_j \leq w_j$.

If $0.25 \leq p^{\max Y}_j < \frac{y_{j-1}}{2}$ holds, the algorithm updates $W$ to contain only $\{j^{\max Y}\}$, i.e $w_j = p^{\max Y}_j$, and in this case we get $\gamma_j < 0.25 \leq w_j < \frac{y_{j-1}}{2} < \frac38 < 0.75\cdot p_j$.

Finally, if $p^{\max Y}_j < 0.25$, the algorithm sorts the jobs of $Y$ by non-increasing size and updates $W$ to be the minimum length prefix of the sorted list with total size at least $0.25$. Since $p^{\max Y}_j < 0.25$ (so all jobs of $Y_{j-1}$ are smaller than $0.25$),  while $y_{j-1}>0.5$, the algorithm initially updates $W$ to be of total size in $[0.25,0.5)$. Since $0.5 < y_{j-1} < 0.75$, the total size of the jobs of the complement set is smaller than $0.5$. Therefore, no matter whether $W$ is the original one or the complement, both $w_j<0.5$ and $y_{j-1}-w_j<0.5$ will hold.

If $w_j \leq 0.75\cdot p_j$, the set $W$ is not modified further, and we are done also with respect to this property of $W$. Otherwise, the algorithm updates $W$ to be its complement set with respect to $Y$. Since $y_{j-1}<0.75\leq \frac{3p_j}2$, for the partition of $Y_{j-1}$, at least one of the two subsets has total size at most $\frac{3p_j}4$, and therefore the total size of the complement set is smaller than $0.75\cdot p_j$, and for the final and possibly modified set $W$, we get $w_j \leq 0.75 \cdot p_j$. On the other hand $\gamma_j= y_{j-1}+p_j-1.25  < (0.5 + w_j) +p_j-1.25 = w_j +p_j -0.75$. Since $p_j < 0.75$ we get $\gamma_j < w_j$.

To prove the last part, note that $j$ will be assigned to machine $m_2$, and recall that we have $\gamma_j \leq w_j \leq 0.75\cdot p_j$ where the algorithm migrates $W$ to machine $m_1$. Thus, after assigning $j$ we get $y_j = p_j + y_{j-1} - w_j \leq p_j + y_{j-1} - \gamma_j = 1.25$.

As for the load of the first machine, we have $y_{j-1}+p_j>1.25$, and therefore the total size of jobs whose GoS is $1$ is below $0.75$. We have $z_{j-1}+y_{j-1}+p_j \leq 2-\Lambda$. If the first machine only receives jobs of GoS $1$ after $j$ was assigned, its load is at most $\Lambda+z_{j-1}+w_j \leq 2-(y_{j-1}+p_j)+w_j < 2- 1.25 +0.5=1.25$.
\end{proof}

\begin{lemma} \label{lB:5}
For every input, the algorithm enters step $5$ at most once. And if the algorithm enters step $5$, all further jobs will be assigned in steps $2$ and $3$.
\end{lemma}
\begin{proof}
If the algorithm does not apply step $5$, we are done. Otherwise, let $j$ be the first job in which the algorithm applies step $5$. Since previous steps were not applied for $j$, we have $0.5< p_j < 0.75$ and $y_{j-1} +p_j>1.25$. If $p_j+p^{\max Y}_j >1.25$, it holds that $p^{\max Y}_j >0.5$. Therefore, since there are already two jobs of sizes above $0.5$, any job that arrives after $j$ has a processing time of at most $0.5$, which means the algorithm will assign all future jobs in steps $2$ and $3$ (see Lemma \ref{lB:2}). If $p_j+p^{\max Y}_j \leq 1.25$, then by Lemma \ref{lB:4}, $j$ will be assigned to machine $m_2$, and $w_j <0.5$, so we get $y_j= y_{j-1} +p_j -w_j \geq 1.25 - 0.5 = 0.75$. Therefore, the algorithm will assign all further jobs in step $2$.
\end{proof}

\begin{corollary} \label{cB:1}
For every input, the algorithm will not enter both steps $4$ and $5$.
\end{corollary}
\begin{proof}
If the algorithm assigns $j$ in step $4$, then by Lemma \ref{lB:3} it assigns all further jobs in steps $2$ and $3$. If the algorithm assigns $j$ in step $5$, the situation is similar by Lemma \ref{lB:5}.
Thus, after one of steps $4$ and $5$ is applied, none of these steps is applied again.
\end{proof}

\begin{lemma} \label{lB:6}
After assigning $j$ in step $4$, the makespan is at most $1.25$. The load of $m_1$ will not exceed $1.25$ even if it receives additional jobs of GoS $1$ in the future.
\end{lemma}
\begin{proof}
By Corollary \ref{cB:1} the algorithm never entered step $4$ or $5$ before $j$. Therefore, machine $m_1$ has no jobs of GoS $2$. This holds since jobs of GoS $2$ are not assigned to $m_1$ in step $3$, and for step $2$, once such a job is assigned to $m_1$, all future jobs (for both hierarchies) will be assigned in the same step, and the algorithm would not have reached step $4$.

Since previous steps were not applied but step $4$ is applied, we have $p_j\geq 0.75$, $y_{j-1}<0.75$, and $y_{j-1}+p_j>1.25$.

Now, if $j$ is assigned to machine $m_2$, we have a subset $W$ such that $y_{j-1}+p_j-1.25 = \gamma_j \leq w_j \leq 0.75\cdot p_j$ that the algorithm migrated it to machine $m_1$, so we get $0.75 \leq p_j \leq y_j = y_{j-1}+p_j-w_j \leq 1.25$, and $x_j+z_j \leq 2- y_j \leq 1.25$.

If $j$ is assigned to machine $m_1$, it holds that  $y_j = y_{j-1} < 0.75$, and we have $y_{j-1}-w_j+p_j>1.25$. If $W$ is empty, since $p_j\leq 1$, the set $Y_{j-1}$ is non-empty and $p^{\max Y}_j > 0.75 \cdot p_j$ holds. Otherwise, by Lemma \ref{lB:1}, $Y_{j-1}$ contains a pair of jobs whose total processing time is greater than $0.75\cdot p_j$, and the total size of each of them together with $j$ is above $1$. Thus, in any optimal solution, one machine has $j$, and the other machine has jobs of GoS $2$ with total size above $0.75\cdot p_j$, and therefore each machine has jobs of total size above $0.75 \cdot p_j$ of GoS $2$.
This implies that the total size of jobs of GoS $1$ is at most $1-0.75\cdot p_j$, since only one machine may have such jobs, and thus  $\Lambda \leq 1-0.75\cdot p_j$. As long as machine $1$ does not get other jobs of GoS $2$, it has jobs of GoS $1$ and the set $Z_j=\{j\}$, and its load is at most $p_j+\Lambda \leq 1-0.75\cdot p_j +p_j \leq 1.25$.
\end{proof}

\begin{lemma} \label{lB:7}
In step $5$, if $p_j + p^{\max Y}_j > 1.25$, then after assigning $j$ the makespan is at most $1.25$. If afterwards $m_1$ receives only jobs of GoS $1$, then its load will remain at most $1.25$.
\end{lemma}
\begin{proof}
Similarly to Lemma \ref{lB:6}, by Corollary \ref{cB:1} the algorithm never entered step $4$ or $5$ before $j$, and machine $m_1$ has only jobs of GoS $1$ (it may be empty). Since the largest job of $Y_{j-1}$ and $j$ are assigned to different machines in any optimal solution, the total size of jobs whose GoS is $1$ is at most $\max\{1-p_j,1-p^{\max Y}_j\}$, and this is an upper bound on $\Lambda$, which we use to denote the total size of jobs of GoS $1$ in the entire input. Once again $y_j=y_{j-1}<0.75$, so it remains to bound the load of the first machine, including all jobs of GoS $1$.

We have $\Lambda+z_j=\Lambda+p_j$. If $\max\{1-p_j,1-p^{\max Y}_j\}=1-p_j$, we get $\Lambda+z_j\leq 1$. Otherwise, $\Lambda+z_j \leq 1-p^{\max Y}_j+p_j< 2\cdot p_j-0.25$. In this case $p_j\leq 0.75$, and therefore $\Lambda+z_j<1.25$.
\end{proof}

\begin{theorem}\label{theoB}
The competitive ratio for algorithm \textit{B} is at most $1.25$.
\end{theorem}
\begin{proof}
Since the optimal cost is $1$, we show by induction that the makespan never exceeds $1.25$. The base case is before any job was assigned. We assume that just before job $j$ (for $j\geq 1$) is assigned, the makespan does not exceed $1.25$. We already showed that the makespan will not exceed $1.25$ after an assignment by step $4$ or $5$.
If $j$ is assigned in step 2 and $y_{j-1}\geq 0.75$, we have $y_j=y_{j-1}\leq 1.25$ by the induction hypothesis and $x_j+z_j \leq 2-y_j \leq 1.25$. If $j$ is assigned by step $3$, the property is satisfied by induction and by the condition of this case. We are left with the case that $j$ is assigned by step $2$ while $y_{j-1}<0.75$, in which case $g_j=1$ holds. If $z_j=0$, we get $x_j\leq 1$, since the total size of jobs of GoS $1$ does not exceed $1$, so we assume that $z_j>0$. This means that one of steps $4,5$ was applied in the past. In all cases for steps $4$ and $5$ we saw that as long as $m_1$ only receives jobs of GoS $1$, its load remains no larger than $1.25$. Thus, we consider the case that it received such a job, which could only happen in step $2$. However, when only steps $2$ and $3$ are applied, the total size of jobs of $Y$ cannot decrease, and when a job whose GoS was assigned to $m_1$ later than the application of step $4$ or step $5$, this means that the total size for $Y$ was at least $0.75$. Thus, $y_{j-1}\geq 0.75$, contradicting the condition of the case.
\end{proof}

\begin{lemma}
The migration factor of the algorithm is at most $\frac 34$.
\end{lemma}
\begin{proof}
The only two steps where jobs may be migrated, step $4$ and $5$. In step $4$, $W$ is migrated, and it is selected such that its total size does not exceed $0.75\cdot p_j$. For step $5$, a similar property was proved in Lemma \ref{lB:4}.
\end{proof}

\section{The case $\boldsymbol{\frac{1}{2} \leq M < \frac{3}{4}}$}

In this section, we prove a tight bound of $2-M$ for the case $M\in [\frac 12,\frac 34)$. In the end of this analysis, we will show a lower bound on the competitive ratio for any algorithm. The upper bound on the competitive ratio is shown by presenting two algorithms (C and D)  defined over two different domains of the migration factor of this case. The idea behind both algorithms is similar to the previous algorithms, but still we present two distinct algorithms because the way in which we determine $W$ (the job sets selected from the second machine) when it needs to migrate jobs, to maintain a completion processing time in the two machine that is not greater than $2-M$.
In the domain $\frac{1}{2} \leq M < \frac{2}{3}$, the $Algorithm \; C$ finds a subset $W$ of $Y$, whose migration will fulfill the required.
In domain $\frac{2}{3} \leq M < \frac{3}{4}$, the $Algorithm \; D$ finds a subset $W$ of $Y$ also, but the two algorithms differ in size of this subset, which is stemmed from the property that the migration factor is larger and the competitive ratio is smaller. We present the proof of the bounds of the competitive ratios, and the migration factors, after presenting them. We begin with the $Algorithm \; C$ with a migration factor in the interval $[\frac{1}{2}, \frac{2}{3})$:

\subsection{An algorithm for the case $\boldsymbol{\frac12 \leq M < \frac23}$}
\textbf{Algorithm \textit{C}}

Let $X=\emptyset,\ Y=\emptyset,\ Z=\emptyset,\ W=\emptyset;$

Repeat until all jobs have been assigned:
\begin{enumerate}
\item Receive job $j$ with $p_j$ and $g_j$; \item If $g_j=1 \ or \ y_{j-1} \geq M$ holds,
 assign $j$ to the first machine and update: $X \leftarrow X\cup \{j\},$ or $Z \leftarrow Z\cup \{j\},$ respectively.

            return to step 1.
\item If $y_{j-1}+p_j \leq 2-M$ holds, assign $j$ to the second machine and update: $Y\leftarrow Y\cup \{j\},$

            return to step 1.
\item If $p^{\max{Y}}_j>M\cdot p_j$ holds, assign $j$ to the first machine and update: $Z \leftarrow Z\cup \{j\}$.

            return to step 1.
\item Sort the jobs in $Y$ by non-increasing size, and update $W$ to be the minimum length prefix of the sorted list with total size at least $p_j+y_{j-1}-(2-M)$ (or let $W=Y$ if not such subset exists).
Migrate all jobs of $W$ to the first machine, and assign $j$ to the second machine, and update: $Y\leftarrow (Y \setminus W)\cup \{j\}$ and $Z\leftarrow Z\cup W$.

            return to step 1.
\end{enumerate}

The idea in $Algorithm \; C$ in the first steps is similar to the idea of its predecessors. It handles small jobs and ensures that the load of second machine load is at most $2-M$. In the last steps, step $4$ and step $5$, $Algorithm \; C$ relies on the fact that in each input there are no three jobs, where the total processing time of each pair of them is more than $1$. Similarly to $Algorithm \; B$, we prove that $Algorithm \; C$ may apply  step $4$ or $5$, and not both, and such a step (one of these two steps) is applied at most once for every input. Then, we prove that if $Algorithm \; C$ enters step $5$, then it will always be able to find a subset $W$ of $Y$ such that $\gamma_j \leq w_j \leq M\cdot p_j$, where $\gamma_j$ in the case $\frac{1}{2} \leq M < \frac{3}{4}$  is equal to $y_{j-1} + p_j -(2-M)$.
The value  $\gamma_j$ is a lower bound on the total size that has to be migrated from the second machine to obtain $y_j\leq 2-M$ in order to allow the option of scheduling $j$ on the second machine.

Note that the if the algorithm reaches step $4$ (and applied this step or step $5$), it holds that $y_{j-1} < M$ and $y_{j-1} +p_j > 2-M$, which implies that $p_j > 2-2\cdot M \geq \frac 23$. In this case, we will have $M\cdot p_j \geq M(2-2M) > \frac 49$. Note that $y_{j-1} +p_j > 2-M$ implies that the set $W$ is selected such that its total size at least $p_j+y_{j-1}-(2-M)$. We also have $p^{\max{Y}}_j \leq M\cdot p_j$ if step $5$ is applied.

\begin{lemma} \label{lC:1}
In step $5$, the algorithm always manages to update $W$ such that $\gamma_j \leq w_j \leq M\cdot p_j$.
\end{lemma}
\begin{proof}
Suppose that the algorithm enters step $5$. The selected set has total size at least $\gamma_j=p_j+y_{j-1}-(2-M)$ since such a subset exists, and therefore $\gamma_j \leq w_j$ holds. It is left to show that $w_j \leq M\cdot p_j$ holds. By $p^{\max{Y}}_j \leq M\cdot p_j$ , we can consider the following two cases:

\begin{enumerate}
\item $\gamma_j \leq p^{\max{Y}}_j \leq M\cdot p_j$.
\item $p^{\max{Y}}_j < \gamma_j$.
\end{enumerate}

In the first case, $W$ will consist of one job, which is the largest job in $Y$, i.e $\gamma_j \leq w_j=p^{\max{Y}}_j \leq Mp_j$. In the second case $p^{\max{Y}}_j < \gamma_j$, and we use the following properties.

\begin{itemize}
\item $y_{j-1} > \gamma_j$, because $p_j+M-2 \leq 1 +\frac 23-2= -\frac 13$ yields  $\gamma_j < \gamma_j - p_j + (2-M) = y_{j-1}.$
\item $M+2 < 4-2M$, because $M < \frac{2}{3}$.
\item $(2-M)p_j+2y_{j-1} < 2+M$, because $y_{j-1}<M, \ p_j\leq 1$.
\end{itemize}

We get $(2-M)p_j+2y_{j-1} < 4-2M \ \Leftrightarrow \ 2(p_j+y_{j-1}-2+M) < Mp_j \ \Leftrightarrow \ 2\gamma_j < Mp_j$. Every job scheduled on the second machine (any job of $Y_{j-1}$) has a size of less than $\gamma_j$, (because $p^{\max{Y}}_j < \gamma_j$). Therefore, a minimum length prefix of the sorted list with total size at least $\gamma_j$ will have a total size of at most $2\gamma_j < Mp_j$.
\end{proof}

\begin{lemma}\label{lC:2}
For any job $j$ satisfying $p_j\leq 2-2M$, $j$ is assigned in step $2$ or step $3$. Once the property $y_j \geq M$ holds for some job $j$, all jobs arriving after $j$ will be scheduled in step $2$.
\end{lemma}
\begin{proof}
The proof is the same as the Lemma \ref{lB:2} proof taking into account the different parameters.

If $g_j=1$, $j$ is assigned in step $2$. If prior to the arrival of $j$ it holds that $y_{j-1} \geq M$, it is also assigned in step $2$. Thus, we are left with the case $g_j=2$ and $y_{j-1}<M$. In this case, $y_{j-1}+p_j < 2-M$, and therefore $j$ is assigned in step $3$.
\end{proof}

\begin{lemma} \label{lC:3}
For every input, the algorithm enters the union of steps $4$ and $5$ at most once.
\end{lemma}
\begin{proof}
If none of these steps is applied, we are done. Otherwise, let $j$ be the first job for which step $3$ is considered but not applied. Since steps $2$ and $3$ were not applied, it holds that $p_j > 2-2M$. If the algorithm enters step $4$, i.e. $p^{\max{Y}}_j > M\cdot p_j > M(2-2M)$. We have $(2-2M)+M(2-2M)>1$ and $2(2-2M)>1$, since $(2-2M)+M(2-2M)=2(1-M)(1+M)=2(1-M^2)\geq 2(1-4/9)>1$ and $M<1$. Therefore, later in the input no job with processing time more than $2-2M$ will arrive (because a machine that will get two of the three largest  will have load above $1$), so by Lemma \ref{lC:2}, the algorithm will not reach step $4$ or step $5$. If $p^{\max{Y}}_j \leq M\cdot p_j$, then $j$ will be assigned to machine $m_2$, and in this case we get $y_j \geq p_j > 2-2M \geq M$. Therefore, by Lemma \ref{lC:2} all jobs that arrive after $j$ will be assigned to the first machine by step $2$.
\end{proof}

\begin{lemma} \label{lC:4}
If job $j$  is assigned in step $4$ or step $5$, after assigning $j$ the makespan is at most $2-M$. If $m_1$ will only get jobs of GoS $1$ afterwards, its load will never exceed $2-M$.
\end{lemma}
\begin{proof}
In this case $j$ has GoS $2$, and it holds that $y_{j-1} < M$ and $y_{j-1} +p_j > 2-M$, since the conditions of steps $2$ and $3$ do not hold. By Lemma \ref{lC:3}, the algorithm never entered step $4$ or $5$ before $j$. Therefore, machine $m_1$ has only jobs of GoS $1$ (or it may be empty), by Lemma \ref{lC:2}. If $p^{\max{Y}}_j>M\cdot p_j$, we have two jobs whose total size is above $1$ which must be assigned to different machines in an optimal solution. Thus, and the total size of jobs with GoS $1$ (recall that it is denoted by $\Lambda$) is at most $\max\{1-p_j,1-p^{\max Y}_j\}$. Step $4$ is applied, and since $y_j=y_{j-1}<M<2-M$, we analyze $m_1$. We have $\Lambda+z_j=\Lambda+p_j$. If $p_j\leq p^{\max Y}_j$, we get $\Lambda+p_j\leq 1$. Otherwise, $\Lambda+p_j\leq 1-p^{\max Y}_j+p_j <1+p_j-M\cdot p_j=1+(1-M)p_j\leq 2-M$, since $p_j\leq 1$.

If $p^{\max{Y}}_j \leq M\cdot p_j$, step $5$ is applied. The algorithm will assign $j$ to machine $m_2$, and by Lemma \ref{lC:1} it will migrate a subset $W$ of size at least $\gamma_j$, so we get: $y_j=y_{j-1} - \gamma_j +p_j \leq 2-M$, and $x_j+z_j \leq 2 - y_j \leq 2- p_j < 2- (2-2M) = 2M < 2-M$. Since steps $4$ and $5$ will not be applied again, the set $Y$ will always have job $j$, so the load of $m_1$ will never exceed $2-p_j < 2-M$ (no matter which jobs it will receive).
\end{proof}

\begin{theorem}
The competitive ratio for algorithm \textit{C} is at most $2-M$.
\end{theorem}
\begin{proof}
We use induction again. We saw that an assignment in steps $4$ or step $5$ is valid in terms of the competitive ratio. So is an assignment in step $3$, and an assignment in step $2$ due to a load of $M$ or more for $m_2$. Once again we are left with the case where a job $j$ with $g_j=1$ is assigned to $m_1$. As before, if steps $4$ and $5$ were never applied, $m_1$ has no jobs of GoS $2$, and if one of these steps was applied but all jobs assigned to $m_1$ later have GoS $1$, the load will not exceed $2-M$. Otherwise, $m_2$ already has load of $M$ or more, so $m_1$ will have a load of at most $2-M$.
\end{proof}

\begin{lemma} \label{lC:5}
The migration factor of the algorithm is at most $M$.
\end{lemma}
\begin{proof}
The only step where jobs may be migrated is step $5$. By Lemma \ref{lC:1}, the algorithm always manages to update $W$ such that $w_j \leq M\cdot p_j$, as required.
\end{proof}

\subsection{An algorithm for the case $\boldsymbol{\frac 23 \leq M < \frac 34}$}

We consider the $Algorithm \ D$ with migration factor in the interval $[\frac 23, \frac 34)$. The idea of $Algorithm \ D$ is the same idea that we reviewed in this article, The algorithm tries to load machine $m_2$ until the load reaches the interval $[M,\ 2-M]$ where it keeps the load on machine $m_1$ to be at most $2-M$. The migration is performed in this algorithm only from machine $m_2$ to machine $m_1$. The manner of selecting subset $W$ is similar to its predecessors, but the selection analysis is more complex. In analyzing subset $W$ selection in $Algorithm \ D$ we not only look at the big largest job in $Y$, but also at the second largest job. Now, we present the algorithm and analyze it.

\textbf{Algorithm \textit{E}}

Let $X=\emptyset,\ Y=\emptyset, \ Z=\emptyset;$

Repeat until all jobs have been assigned:
\begin{enumerate}

\item Receive job $j$ with $p_j$ and $g_j$; \item If $g_j=1 \ or \ y_{j-1} \geq M$ holds, assign $j$ to the first machine and update: $X \leftarrow X\cup \{j\},$ or $Z \leftarrow Z\cup \{j\},$ respectively.

            return to step 1.
\item If $y_{j-1}+p_j \leq 2-M$ holds, assign $j$ to the second machine and update:
$Y\leftarrow Y\cup \{j\}$,

            return to step 1.
\item If $p_j \geq M$ holds, sort the jobs in $Y$ by non-increasing size, and update $W$ to be the maximum length prefix of the sorted list with total size at most $M\cdot p_j$ (which may be empty).
\begin{itemize}
\item If $y_{j-1}-w_j+p_j > 2-M$ holds, assign $j$ to the first machine and update: $Z\leftarrow Z\cup\{j\}$.
\item Otherwise, migrate the jobs of $W$ to the first machine, assign $j$ to the second machine and update: $Y\leftarrow (Y \setminus W) \cup \{j\}$ and $Z\leftarrow Z\cup W$.\end{itemize}
            return to step 1.

\item Sort the jobs of $Y$ by non-increasing size and update $W$ to be the minimum length prefix of the sorted list with total size at least $\frac{M}{3}$, or the entire set if no such a prefix exists.

\begin{itemize}
\item If $w_j > \min\{\frac{2M}3, M\cdot p_j\}$ holds, update $W\leftarrow Y \setminus W$.
\item If $y_{j-1}-w_j+p_j > 2-M$ holds, assign $j$ to the first machine and update: $Z\leftarrow Z\cup\{j\}$.
\item Otherwise, migrate the jobs of $W$ to the first machine, assign $j$ to the second machine and update: $Y\leftarrow (Y \setminus W) \cup \{j\}$ and $Z\leftarrow Z\cup W$.\end{itemize}
            return to step 1.
\end{enumerate}

First, we analyze step $4$ in the algorithm, and the first job $j$ assigned by this step, if it exists. By the condition of assignment by step $4$, $p_j \geq M$ holds. We divide the analysis into two cases:

\begin{enumerate}
\item Job $j$ is assigned to machine $m_2$.
\item Job $j$ is assigned to machine $m_1$.
\end{enumerate}

In both cases we prove that the makespan is at most $2-M$. We start with the first case.

\begin{lemma}
\label{lD:1}
In step $4$, if the algorithm assigns $j$ to the second machine, then the final load of the two machines will be at most $2-M$.
\end{lemma}
\begin{proof}
By the condition of assignment of $j$ to $m_2$, we have $y_j=y_{j-1}-w_j+p_j\leq 2-M$. Due to the size of $j$, we have $y_j\geq p_j \geq M$. From this moment onwards, the new jobs will be scheduled on the first machine by step $2$, so the final load of machine $m_2$ will be at least $M$ and at most $2-M$, and the final load of machine $m_1$ will be at most $2-M$.
\end{proof}

\begin{lemma}
\label{lD:2}
If $j$ is assigned to $m_1$, it holds that $|W|\leq 1$ and $|W|<|Y|$.
\end{lemma}
\begin{proof}
If $W=Y$, we have $w_j=y_{j-1}$ and $y_{j-1}-w_y+p_j =p_j \leq 1 \leq 2-M$. Thus, in the case $W=Y$,  $j$ is assigned to $m_2$.

If $W\geq 2$, while $|Y|>|W|$, letting $k=|W|$ (where $k\geq 2$), we have the following. The first job of $Y$ (in the sorted order) that is not in $W$ is not larger than any job of $W$, so $w_j\geq k \cdot x$, where $x$ is the size of that job. Since  $w_j+x > M\cdot p_j$ and $x \leq \frac{w_j}k$, we have $\frac{k+1}k\cdot w_j > M \cdot p_j$ and by $k \geq 2$, we get $w_j > \frac 23 \cdot M \cdot p_j$.

On the other hand, if we show that $w_j \geq 2M-2+p_j$, we will have $y_{j-1}-w_j+p_j \leq y_{j-1}-2M+2 \leq 2-M$, since $y_{j-1}< M$. This would imply that $j$ is assigned to $m_2$. It if sufficient to show that $\frac 23 \cdot M \cdot p_j \geq 2M-2+p_j$, and by rearranging, $(1-\frac {2M}3) \cdot  p_j \leq 2-2M$. Since $p_j\leq 1$ and $1-\frac {2M}3 > \frac 12>0$ for $M<\frac 34$, it is sufficient to show that $1-\frac {2M}3 \leq 2-2M$, or equivalently, $\frac{4M}3\leq 1$, which holds since $M\leq \frac 34$.
\end{proof}

Recall that $\gamma_j=y_{j-1} + p_j -(2-M)$ in this case.

\begin{lemma}
\label{lD:3}
Assume that $j$ is assigned in step $4$. If it holds that $p^{\max{Y}}_j <\gamma_j$, and $p^{\max{Y,2}}_j+p^{\max{Y}}_j > M\cdot p_j$, then $p_j+p^{\max{Y,2}}_j > 1$ holds (and $p_j+p^{\max{Y}}_j > 1$ holds as well). Moreover, if $p^{\max{Y}}_j > M\cdot p_j$, then $p_j+p^{\max{Y}}_j > 1$ holds. In both situations, the total size of all jobs of GoS 1 is at most $1-M\cdot p_j$.

\end{lemma}
\begin{proof}
We start with the first property. Let us assume (by contradiction) that the conditions hold but $p_j+p^{\max{Y,2}}_j \leq 1$. We get (by  $y_{j-1}<M$ which holds since step $2$ was not applied for $j$):
$$ M\cdot p_j < p^{\max{Y,2}}_j+p^{\max{Y}}_j <1-p_j+\gamma_j = y_{j-1}+M-1 < 2M-1 \ ,$$
i.e. $M\cdot p_j < 2M-1$. Therefore $M^2 < 2M-1$, which holds since $p_j\geq M$ (since $j$ is assigned in step $4$). A contradiction, because $M^2-2M+1 = (M-1)^2 \geq 0$. Thus, we find that $p_j+p^{\max{Y,2}}_j > 1$ holds, and it also holds that $p_j+p^{\max{Y}}_j>1$,
because $p^{\max{Y}}_j\geq p^{\max{Y,2}}_j$.

If we have $p^{\max{Y}}_j > M\cdot p_j$, we get $p_j+p^{\max{Y}}_j > (M+1)\cdot p_j\geq M(M+1)\geq \frac 23 \cdot \frac 53>1 $.

In both cases, there is a job or even a pair of jobs that cannot be assigned to the machine of $j$ in any optimal solution. If $j$ is assigned to $m_1$ in such a solution, the total size of other jobs of $m_1$ is at most $1-p_j \leq 1-M\cdot p_j$ since $M<1$, and there are no other jobs whose GoS is $1$. Otherwise, the upper bound on jobs of GoS $1$ is $1-M\cdot p_j$ due to the size of the largest job or two largest jobs of $Y$ that cannot be assigned to $m_2$ in this case (in the first situation, this is based on a pair of jobs and in the second situation it is based on one job).
\end{proof}

\begin{lemma}
\label{lD:4}
In step $4$, if the algorithm assigns $j$ to the first machine, then the loads of the two machines will not exceed $2-M$,and if $m_1$ will only get jobs of GoS $1$ afterwards, its final load will be at most $2-M$. Moreover, further jobs will assigned by steps $2$ and $3$, and the final makespan will not exceed $2-M$.
\end{lemma}
\begin{proof}
By Lemma \ref{lD:2}, we have $|W|\leq 1$, and $Y$ has at least one other job. If $W$ is empty, this means that already $p^{\max{Y}}_j > M\cdot p_j$ holds. In this case, the load of $m_1$ is at most the total size of jobs of GoS $1$ plus $p_j$, which is at most $1-M\cdot p_j+p_j$ by Lemma \ref{lD:3}. We have $1-M\cdot p_j+p_j=1+p_j(1-M)\leq 2-M$ by $p_j\leq 1$. The upper bound on the load of $m_1$ is valid not only after $j$ is assigned to it but also as long as this machine receives only jobs of GoS $1$.

If $W$ consists of one job, then $p^{\max{Y,2}}_j+p^{\max{Y}}_j > M\cdot p_j$ and due to the assignment to $m_1$, it holds that $p^{\max{Y}}_j = w_j <\gamma_j$. By  Lemma \ref{lD:3}, the analysis for $m_1$ is as in the previous case. This completes the analysis for $m_1$, and the load of $m_2$ is unchanged in the assignment of $j$, and it remains below $M<1$.
Since a job of size $M$ was just assigned to $m_1$, all further jobs of GoS $2$ can be assigned in step 3 if they are not assigned in step $2$, and such an assignment will not increase the loads above $2-M$. An assignment of a job of GoS $2$ in step $2$ means that the load of $m_2$ is already at least $M$, so after such an assignment is performed, all further jobs are assigned in step $2$, keeping the load of $m_1$ not larger than $2-M$ (and keeping the load of $m_2$ unchanged).
\end{proof}

For the analysis of an assignment of a job $j$ in step $5$, the case where in particular it holds that $p_j < M$, we will first prove some auxiliary lemmas. Note that if step $5$ is applied (and the previous steps were not), we have $p_j < M$, and $y_{j-1}+p_j > 2-M$, which implies $y_j > 2-2M > \frac{2M}{3} > \frac 12$ by $M< \frac 34$, and therefore the set $W$ indeed has total size above $\frac M3$.

\begin{lemma} \label{lD:5}
In step $5$, given a set of jobs which is a subset of $\{1,2,\ldots,j-1\}$, whose total size is at most $\frac{M}{2}$, it is possible to migrate these jobs.
\end{lemma}
\begin{proof}
Since $j$ is not assigned in steps $2$ and $3$, we have $p_j > 2-2M$. The proof follows immediately from the inequality: $p_j > 2-2M > \frac{1}{2}$, so $M\cdot p_j \geq \frac M2$.
\end{proof}

\begin{lemma}
\label{lD:6}
In step $5$, $y_{j-1} > \frac{2M}{3}$ and $p_j > \frac{2M}{3}$ hold.
\end{lemma}
\begin{proof}
In this case $y_{j-1} < M$ and $p_j < M$ hold since previous steps were not applied. If one of them has value which is not greater than $\frac{2M}{3}$, we get
$y_{j-1}+p_j < \frac{2M}{3} + M = \frac{5M}{3} < \frac{5}{3} \cdot \frac{3}{4} = 1.25,$ a contradiction, because the algorithm considered step $3$ for $j$ and did not apply it, i.e. $y_{j-1}+p_j > 2-M > 1.25$.
\end{proof}

\begin{lemma}
\label{lD:7}
In step $5$, if the algorithm migrated jobs from $m_2$, and scheduled $j$ on $m_2$, and we have $\frac{M}{3} \leq y_j-p_j \leq \frac{2M}{3}$, then $M \leq y_j \leq 2-M$.
\end{lemma}
\begin{proof}
It holds that $y_j \leq p_j + \frac{2M}3 \leq 2-M$, since $p_j\leq M \leq \frac 34$.

On the other hand $y_j \geq \frac M3 + p_j > \frac M3 + (2-2M)> M$, because $\frac{8M}{3} < 2$ (for $M < \frac 34$).
\end{proof}

\begin{lemma}
\label{lD:8}
In step $5$, if the algorithm migrates jobs from $m_2$, where the total size of these jobs is in the interval $[\frac{M}{3},\frac{2M}{3}]$, and it schedules $j$ on $m_2$, then as a result we get $M \leq y_j \leq 2-M$.
\end{lemma}
\begin{proof}
If the algorithm migrates a total size which is no less than $\frac{M}{3}$, then we get $y_j \leq p_j +y_{j-1} - \frac{M}{3} \leq 2M-\frac{M}{3} = \frac{5M}{3} < 2-M$. If the algorithm migrates a total size of jobs which is no more than $\frac{2M}{3}$, then we get $y_j \geq p_j+y_{j-1}-\frac{2M}{3} \geq 2-M- \frac{2M}{3} > M,$ because $\frac{8M}3 < 2$ for $M < \frac 34$.
\end{proof}

\begin{lemma}
\label{lD:9}
In step $5$, if the algorithm assigns $j$ to machine $m_2$, then the final makespan will be at most $2-M$, and the migration factor is not violated.
\end{lemma}
\begin{proof}
In this case the algorithm checked if migration of $W$ would violate the migration factor, and replaced it with $Y\setminus W$ if it would. By Lemma \ref{lD:5}, and since $y_{j-1}<M$, at least one of these sets can be migrated (since the total size of the set with the smallest total size out of $W$ and $Y\setminus W$ is below $\frac M2$). Thus, if the original $W$ cannot be migrated, its complement can be migrated with respect to the migration factor. Recall that $y_{j-1}>\frac{2M}{3}$ by Lemma \ref{lD:6}, and the set $W$ is first defined to have a total size of at least $\frac {M}3$. Consider the following cases:

Case 1: $p^{\max{Y}}_j \leq \frac{2M}{3}$. In this case, the algorithm initially updates $W$ such that $\frac{M}{3}\leq w_j \leq \frac{2M}{3}$. The lower  bound is based on the total size. For the upper bound, we find that if the largest job of $Y$ is of size at least $\frac M3$, then $W$ initially consists of this job. Otherwise all jobs have sizes below $\frac M3$, and when the last job is added, the total size is below $\frac M3$.

If $w_j \leq Mp_j$, the algorithm migrates $W$ to $m_1$ and schedules $j$ on $m_2$, so by Lemma \ref{lD:8} we get $M \leq y_j \leq 2-M$. Otherwise, the algorithm updates $W$ to be $Y \setminus W$, migrates it to $m_1$ and schedules $j$ on $m_2$. In the last case, we get $M \leq y_j \leq 2-M$ by Lemma \ref{lD:7}, since $y_j-p_j$ is exactly the total size of jobs of $W$ before it is replaced with its complement.

Case 2: $p^{\max{Y}}_j > \frac{2M}{3}$, in this case, the algorithm first selects $W$ to be the maximum job in $Y$. After this, the algorithm update $W$ to be $Y \setminus W$ and migrate it to $m_1$ and schedule $j$ on $m_2$, so we get $y_j =p_j+ p^{\max Y}_j > \frac{2M}{3} + \frac{2M}{3} > M$ (see Lemma \ref{lD:6}), and by the condition of assignment of $j$ to $m_2$, we have $y_j= y_{j-1}+p_j-w_j \leq 2-M$.

In the previous cases we get $M \leq y_j \leq 2-M$, so all jobs arriving after $j$ will be scheduled on the first machine by step $2$, and we get: $M \leq T_2 =y_j \leq 2-M$, and $T_1 =2-T_2 \leq 2-M$.
\end{proof}

\begin{corollary} \label{cD:1}
In step $5$, if the algorithm assigns $j$ to machine $m_1$, then $p^{\max{Y}}_j + p_j > 2-M$.
\end{corollary}
\begin{proof}
In this case $p^{\max{Y}}_j > \min\{\frac{2M}3, M\cdot p_j\}$ holds, because if not, the algorithm will always be able to migrate a set $W$ to machine $m_1$ and assign $j$ to machine $m_2$ as we saw in Lemma \ref{lD:9}. Therefore, in this case $w_j = y_{j-1} - p^{\max{Y}}_j < \frac{M}{2}$. Thus, the algorithm assigns $j$ to machine $m_1$ because $y_{j-1}-w_j+p_j > 2-M$, i.e. $p^{\max{Y}}_j +p_j > 2-M$.
\end{proof}

\begin{lemma}
\label{lD:10}
In step $5$, if the algorithm assigns $j$ to machine $m_1$, then the final makespan will be at most $2-M$.
\end{lemma}
\begin{proof}
By Corollary \ref{cD:1} $p^{\max{Y}}_j +p_j > 2-M$, i.e. $p^{\max{Y}}_j +p_j > 1$, and it holds that $y_{j-1} + p_j > 2-M$, since the condition of step $3$ do not hold. Therefore, every job $k$ that arrive after $j$ has $p_k < M$, thus, the algorithm will not enter step $4$ after $j$. Moreover, the algorithm will not enter step $5$ after $j$, because if so then $p_k>2-2M$ holds, i.e. $p^{\max{Y}}_j + p_k > \frac{M}{2}+2-2M>1$ and $p_j + p_k > (2-2M)+(2-2M)>1$.
Therefore, later in the input all jobs will be assigned in step $2$ or $3$ (because a machine that will get two of the three largest will have load above $1$). We have two jobs whose total size is above $1$ which must be assigned to different machines in an optimal solution. Thus, the total size of jobs with GoS $1$ is at most $\max\{1-p_j,1-p^{\max Y}_j\}$. If machine $m_1$ will only receive jobs with GoS $1$ until the end of the input then $T_1 = x_n+p_j$. If $p_j\leq p^{\max Y}_j$, we get $x_j+p_j\leq 1$. Otherwise, $x_j+p_j\leq 1-p^{\max Y}_j + p_j < 1+p_j-(2-M-p_j) < 3M -1 \leq 2-M$, since $p_j < M$ and $M\leq \frac 34$. If machine $m_1$ receive job with GoS $2$ then the load of machine $m_2$ is at least $M$. In both cases we get  $T_1 \leq 2-M$ and $T_2 \leq 2-M$.
\end{proof}

\begin{theorem}
\label{theoD:1}
The competitive ratio for algorithm \textit{D} is at most $2-M$, and the migration factor is not larger than $M$.
\end{theorem}
\begin{proof}
We showed that the competitive ratio is valid if step $4$ or $5$ is reached. Otherwise, if jobs of GoS $2$ are only assigned to $m_2$, we are done, and in particular the load of $m_2$ never exceeds $2-M$. If a job of GoS $2$ is ever assigned to $m_1$, then machine $m_2$ already has a load of at least $M$, and machine $m_1$ will not have a load above $2-M$.

As for the migration factor, we observe that jobs are only migrated in steps $4$ and $5$, while in step $4$ by definition we have that the total size of $W$ does not exceed $M\cdot p_j$. In step $5$, there is migration only in the case that the largest job is sufficiently small to be migrated. If other jobs are migrated, there are two cases. In the case where their total size is at most $\frac{M}2$, we are done by Lemma \ref{lD:5}. The remaining case is the first one. In this case if $W$ is replaced with its complement, the complement has total size at most $M\cdot p_j$ since $y_j \leq M$, so at least one of the two sets has total size no larger than $\frac M2$.
\end{proof}

\subsection{A lower bound on the competitive ratio of the case where $\boldsymbol{\frac{1}{2}\leq M<\frac{3}{4}}$}

\begin{theorem}
\label{ALGALOW0.5} The competitive ratio for every algorithm when $\frac{1}{2}\leq M<\frac{3}{4}$ is at least $2-M$.
\end{theorem}
\begin{proof}
Let $\eps>0$ be a small constant such that $\eps<0.1$, such that $\frac{1}{\eps}$ is an integer. The input starts with the first job $J_1=(M+\eps,2)$ (where $M+\eps<1$). If the algorithm schedules the first job on the first machine, then there is a job of the form $(1,1)$. It will not be possible to migrate of the first job as $\frac{M+\eps}{1}>M$. An optimal solution assigns the second job to $m_1$ and the first job is assigned to $m_2$, and the makespan is $1$. The algorithm will have a load of  $1+M+\eps > 2-M$ for $m_1$.

If the algorithm schedules the first job on the second machine, then the second job is $J_2=(1,2)$, and it  will not allow to migrate the first job since $\frac{M+\eps}{1}>M$. The optimal makespan is $1$ which is achieved by assigning the first two jobs to different machines. If the algorithm schedules the second job on the second machine (so both jobs are assigned to the same machine), we will get a competitive ratio of $1+M+\eps>2-M$ again.

If the algorithm schedules the second job on the first machine, then there is a third job $J_3=(1-M-\eps,1)$. This job cannot cause the migration of $J_2$ or of $J_1$, since it is smaller (as $1-M \leq \frac 12$). Thus, $m_1$ will have both $J_2$ and $J_3$, and a load of $2-M-\eps$. An optimal solution assigns $J_2$ to $m_2$ and the other jobs to $m_1$, for a makespan of $1$. The competitive ratio $2-M-\eps$ tends to $2-M$ by letting $\eps$ tend to zero.
\end{proof}

\section{The case $\boldsymbol{0\leq M<\frac{1}{2}}$ and comments}

In this section we show that the case where $M<\frac 12$ is equivalent to the case without migration, by adapting the lower bound construction. Recall that an algorithm of competitive ratio $1.5$ without migration is known \cite{JS016}.

\begin{theorem}
\label{ALGALOW} The competitive ratio for every algorithm when $M <\frac{1}{2}$ is at least $1.5$.
\end{theorem}
\begin{proof}
The first job is defined by $J_1=(0.5,2)$. If the online algorithm schedules is on the first machine, then there is a job $J_2=(1,1)$. Since $M<\frac 12$, the first job cannot be migrated. An optimal solution schedules the first job on the second machine (and has makespan $1$ due to the second job), while the algorithm schedules both on the first machine, and the competitive ratio is $1.5$. Otherwise, there are two additional jobs: $J_2=(1,2)$ and $J_3=(0.5,1)$, where the sizes are such that no job can cause the migration of another job. An optimal solution schedules the first and third jobs to the first machine, and the second job is scheduled to the second machine, and its makespan is $1$. The algorithm must schedule the third job to the first machine. No matter whether the second job is assigned to the first machine or the second machine, its makespan is $1.5$ and this is also the value of the competitive ratio.
\end{proof}

We summarize the results for all values of $M$ in Table \ref{f:5}.

\begin{figure}[b!]
\begin{center}
\begin{tabular}{c}
\includegraphics[width=140mm]{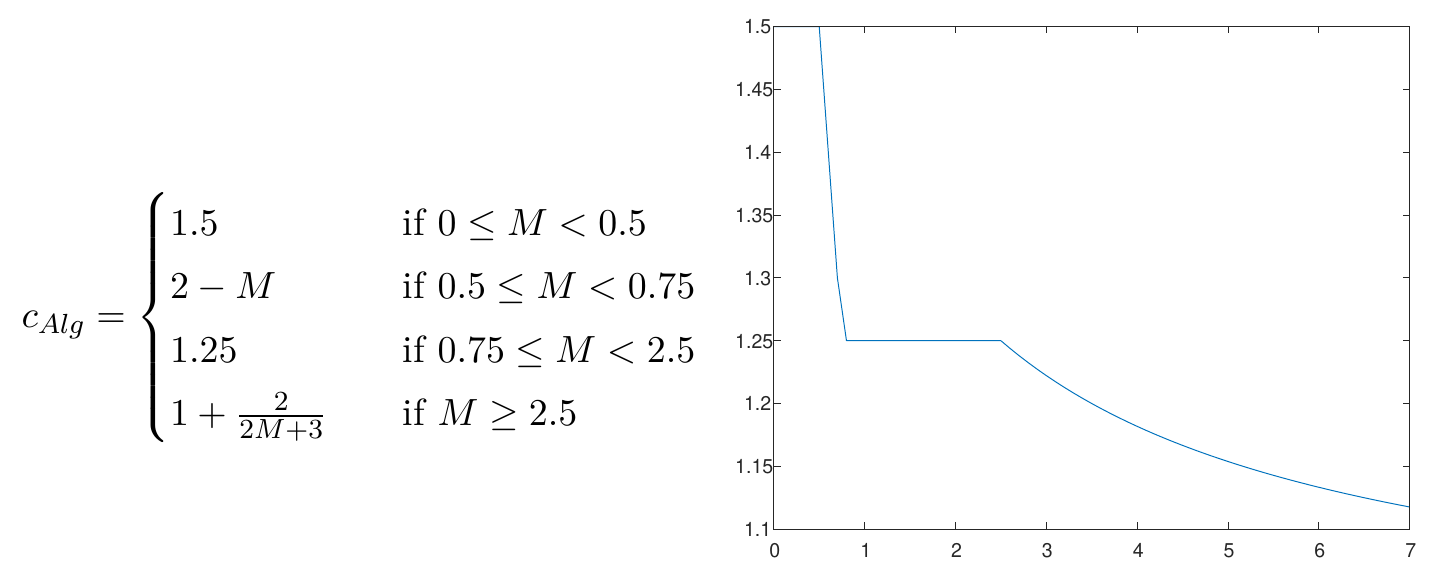}
\end{tabular}
\caption{\small The tight competitive ratio for the semi-online variant with known makespan, as a function of the migration factor $M$. \label{f:5}}
\end{center}
\end{figure}

\subsection{A related variant}
It is known that the following two semi-online models are related. The first one is the one where the optimal makespan is given in advance, and the other one is where the total size of jobs is known in advance. It is known that an algorithm for the latter can be used as an algorithm for the former with the same competitive ratio (and therefore a lower bound on the competitive ratio for the former is a lower bound for the latter). There are specific problems where the results for the two models are similar and problems for which they are different.

We show that the two problems are very different in the sense that for $M$ growing to infinity the variant with known total size has a competitive ratio not tending to $1$ (as the version studied here) but it is bounded away from $1$.

\begin{proposition}
Any algorithm for the semi-online problem where the total processing time is given and any migration factor $M$ has competitive ratio of at least $\frac{\sqrt{33}-1}{4}\approx 1.18614$.
\end{proposition}
\begin{proof}
First, we declare that the total processing time is 2.

The input starts with two large jobs of the form $(\theta,2)$, where $\frac 12< \theta<\frac 23$ is a constant satisfying $4\theta^2+\theta-2=0$, that is, $\theta=\frac{\sqrt{33}-1}{8}\approx 0.59307$. The other jobs will be sufficiently small so that the large jobs cannot be migrated and their initial assignment by an online algorithm is final.

If both large jobs are assigned to $m_2$, the remaining jobs are very small jobs (sand) of total size $2-2\cdot \theta$ with GoS $2$. An optimal solution can split the jobs evenly  such that the makespan is $1$ by assigning one large job each machine with sand of total size $1-\theta$. The algorithm has makespan of at least $2\theta$.

Otherwise, the sand of total size $2-2\theta$ has GoS $1$. In this case, an optimal solution cannot have makespan $1$, and by assigning both large jobs to $m_2$ it has makespan $2\theta$, while the algorithm has to assign all the sand to $m_1$ and its makespan is at least $\theta+(2-2\theta)=2-\theta$.

The competitive ratio is at least $\min\{\frac{2\theta}{1},\frac{2-\theta}{2\theta}\}$. By the equality $\frac{2-\theta}{2\theta}=2\theta$ and using the approximate value of $\theta$, we find that the competitive ratio of any algorithm and any migration factor is $\frac{\sqrt{33}-1}{4}\approx 1.18614$ even if the total job size is given in advance to the algorithm.
\end{proof}

\bibliography{twomach}
\bibliographystyle{abbrv}

\end{document}